\documentclass[12pt]{article}
\usepackage{graphicx}
\usepackage{amssymb}
\usepackage{amsthm}
\usepackage{gastex}
\usepackage{url}
\usepackage[hidelinks,
pdftitle={The Rightmost Equal-Cost Position Problem},
pdfauthor={M. Crochemore, A. Langiu and F. Mignosi},
pdfkeywords={lossless data compression, text compression, parsing problem, optimal parsing, dictionary compression, Lempel-Ziv compression}
]{hyperref}

\theoremstyle{plain}
\newtheorem{theorem}{Theorem}[section]

\newtheorem{proposition}[theorem]{Proposition}
\newtheorem{property}[theorem]{Property}

\theoremstyle{definition}
\newtheorem{definition}{Definition}[section]

\theoremstyle{remark}
\newtheorem*{remark}{Remark}

\author{Maxime Crochemore$^{1,3}$,  Alessio Langiu$^1$ and Filippo Mignosi$^2$ 
\\
\\
$^1$ King's College London, London, UK \\
\textit{\{Maxime.Crochemore,Alessio.Langiu\}@kcl.ac.uk} \\
$^2$ University of L'Aquila, L'Aquila, Italy \\
\textit{Filippo.Mignosi@di.univaq.it} \\
$^3$ Universit\'e Paris-Est, France
}

\title{\Large{\textbf{The Rightmost Equal-Cost Position Problem}}}

\makeatletter

\def\@maketitle{%
  \newpage
  \null
  \vskip 2em%
  \begin{center}%
  \let \footnote \thanks
    {\LARGE \@title \par}%
    \vskip 1.5em%
    {\large
      \lineskip .5em%
      \begin{tabular}[t]{c}%
        \@author
      \end{tabular}\par}%
  \end{center}%
  \par
  \vskip 1.5em}

\makeatother

\begin{document}

\maketitle

\begin{abstract}
LZ77-based compression schemes compress the input text by replacing factors in the text with an encoded reference to a previous occurrence formed by the couple (length, offset). For a given factor, the smallest is the offset, the smallest is the resulting compression ratio. This is optimally achieved by using the rightmost occurrence of a factor in the previous text. 
Given a cost function, for instance the minimum number of bits used to represent an integer, we define the Rightmost Equal-Cost Position $($REP$)$ problem as the problem of finding one of the occurrences of a factor which cost is equal to the cost of the rightmost one. We present the Multi-Layer Suffix Tree data structure that, for a text of length $n$, at any time $i$, it provides REP$($LPF$)$ in constant time, where LPF is the longest previous factor, i.e. the greedy phrase, a reference to the list of REP$(\{$set of prefixes of LPF$\})$ in constant time and REP$(p)$ in time $O(|p| \log \log n)$ for any given pattern $p$. 
\end{abstract}

\section*{Introduction}\label{C:intro}

The Rightmost Equal-Cost Position (REP) is an occurrence, pointed out by its offset, w.r.t. the current position, which reference uses a number of bits that is equal to the number of bits of the offset of the rightmost occurrence of a factor (substring) in the previous text.
This problem mainly comes from the data compression field and particularly from the LZ77 based schemes, where the length and offset pair are used to encode a dictionary phrase.

The foundational Ziv and Lempel LZ77 algorithm \cite{lz77} is the basis of almost all the famous dictionary compressors, like gZip, PkZip, WinZip and 7Zip.
They consider a portion of the previous text as the dictionary, i.e. they use a dynamic dictionary formed by the set of all the factors of the text up to the current position within a sliding window of fixed size.
A dictionary phrase refers to an occurrence of such phrase in the text by using the couple (length, offset), where the offset is the backward offset w.r.t. the current position. Since a phrase is usually repeated more than once along the text and since pointers with smaller offset have usually a smaller representation, the occurrence close to the current position is preferred.

Usually the length and offset pair is encoded by using a variable length code such as Huffman code or Elias Delta code (see for instance Deflate, gZip, PkZip, LZRW4, LZMA and 7zip descriptions in \cite{books/Salomon:2007:DC}).
Assuming that the lengths of encoded values is a monotonic non decreasing function, it is straightforward that the smaller is the value, the smaller is its encoded cost, i.e. the length in bits of the encoded value.
We define a Rightmost Equal-Cost Position (REP) of $w$ as the one of the occurrences of the factor $w$ which cost is the same as the rightmost one, according to a fixed cost function. The REP values can be optimally used to represent a dictionary phrase instead of the rightmost occurrence.
The cost function can be, for instance, the number of bits of the binary representation, i.e. the smaller integer greater than the $\log_2 (x)$, where $x$ is the value we have to encode. Later on in this paper we show how an example using the bit length of the Elias Gamma code as cost function.
Furthermore, in LZ77 based compression, the greedy approach is used to parse the text into phrases, i.e, in an iterative way, the longest match between the dictionary and the forwarding text is chosen.
This is commonly called the greedy phrase.
Some LZ77-based algorithms as Deflate algorithm and the compressors based on them, like gZip and PkZip, use variants of the greedy approach to parse the text. Deflate64 algorithm implemented in WinZip and 7zip, contains some heuristics to parse differently the text in order to improve the compression ratio, but its time complexity was never clearly stated.

The research about parsing optimality for dictionary compression produced in the last decades some noticeable results (see for instance
\cite{DBLP:conf/dcc/CohnK96,DBLP:conf/dcc/Horspool95,cancan,Lempel:2006:AOP:2263298.2268473,DBLP:conf/soda/MatiasS99,DBLP:journals/jacm/StorerS82}).
The greedy parsing is optimal for LZ77 dynamic dictionary schemes in the case of fixed length encoding of length and offset pair values as showed in \cite{DBLP:conf/dcc/CohnK96,DBLP:journals/jacm/StorerS82}, but it is not optimal, in terms of compression performance, when variable length code (VLC) are used to encode the pointers to dictionary phrases.
Recently, some optimal parsing solutions have been presented for dynamic dictionary and VLC, see
\cite{cglmr_iwoca10,cglmr_JDA2011,FerraginaSODA09}.
In \cite{cglmr_JDA2011} presented a solution for LZ77 dictionary and VLC, called Dictionary-Symbolwise Flexible Parsing.
It uses a graph-based model for the parsing problem where each node represent a position in the text and edges represent dictionary phrases. Edges are weighted according to the bit length of the encoded length and offset pair. For details about this algorithm we refer to the original paper \cite{cglmr_JDA2011}.
Unfortunately, this latter paper lacks of a practical solution for the problem of finding the smallest offset (the rightmost position) for a dictionary phrase. Even more, optimal parsing algorithms for the VLC case need to know, for any position in the text, the length and offset pairs of all the prefixes of the Longest Previous Factor (LPF) or, at least, those prefixes which occurrence have a different offset.

\medskip

The main goal of this work is to state the problem about the Rightmost Equal-cost Position (REP) and to present the Multilayer Suffix Tree as an efficient solution.
In Section~\ref{chapter:def} we show how to reduce the problem of finding the rightmost position of a word $w$ to the weaker REP$(w)$ problem.
In Section~\ref{chapter:mlst} we present the new Multilayer Suffix Tree full-text index data structure that uses $O(n \log n)$ amortized building time and linear space.
In Section~\ref{chapter:REP(p)} we show how to solve the REP problem for any pattern $p$ in $O(|p| \log \log n)$ time. 
Finally, in Section~\ref{chapter:REP(LPF)} we show how to
use this new data structure, at any time $i$, to retrieve REP$($LPF$)$ and a pointer to the list of REP$($SPF$)$ in constant time,
where LPF is the longest previous factor starting at position $i$ and SPF is the set of prefixes of LPF, and, furthermore, we show some experimental results.

\section{Definition of the Problem}\label{chapter:def}

Let $Pos(w) \subset \mathbb{N}$ the set of all the occurrences of $w \in \mbox{\emph{Fact}}(T)$ in the text $T\in \Sigma^*$, where $\mbox{\emph{Fact}}(T)$ is the set of the factors of $T$. 
Let $\emph{Offset} (w) \subset \mathbb{N}$ be the set of all the occurrence offsets of  $w \in \mbox{\emph{Fact}}(T)$  in the text $T$, i.e. $x\in \emph{Offset}(w)$ \emph{iff} $x$ is the distance between the position of an occurrence of $w$ and the end of the text $T$. 
For instance, given the text $T=babcabbababb$ of length $|T|=12$ and the factor $w=abb$ of length $|w|=3$, the set of positions of $w$ over $T$ is  \emph{Pos}$(w)=\{4, 9\}$. 
The set of the offsets of $w$ over $T$ is $\emph{Offset} (w)=\{7, 2\}$. 
Notice that $x\in \emph{Offset} (w)$ \emph{iff} exists $y\in  \mbox{\emph{Pos}}(w)$ such that $x = |T| - y - 1$.
Since the offsets are function of positions, there is a bijection between $\mbox{\emph{Pos}}(w)$ and $\emph{Offset}(w)$, for any factor $w$.

Given a number encoding method,  let $\mbox{\emph{Bitlen}}: \mathbb{N} \rightarrow \mathbb{N}$ a function that associates to a number $x$  the length in bit of the encoding of $x$. 
Let us consider the equivalence relation \emph{having equal codeword bit-length} on the set $\emph{Offset} (w)$.  The numbers $x,y\in \emph{Offset} (w)$ are bit-length equivalent \emph{iff} $\mbox{\emph{Bitlen}}(x)=\mbox{\emph{Bitlen}}(y)$. 
Let us notice that the \emph{having equal codeword bit-length} relation induces a partition on $\emph{Offset} (w)$.

\begin{definition}
Given a word $w \in \Sigma^*$, a text $T=a_1 \dots a_n \in \Sigma^*$, the \emph{rightmost} occurrence of $w$ over $T$ is the occurrence of $w$ that appears closest to the end of the text, if $w$ appears at least once in the text $T$, otherwise it is not defined. More formally,
  
$$\mbox{\emph{rightmost}}(w) =
\left\{
	\begin{array}{ll}
		\min\{x\ |\ x\in \emph{Offset} (w)\}  & \mbox{if } \emph{Offset} (w)\neq \emptyset \\
		\mbox{not defined} & \mbox{if } \emph{Offset} (w)= \emptyset
	\end{array}
\right.$$ 

\end{definition}

Let us notice that referring to the rightmost occurrence of a word in a dynamic setting, where the input text is processed left to right, corresponds to referring to the rightmost occurrence over the text \emph{already processed}. 
Indeed, if at a certain algorithm step we have processed the first $i$ symbols of the text, the rightmost occurrence of $w$ is the occurrence of $w$ closest to the position $i$ of the text. 

\begin{definition} Let $\mbox{\emph{rightmost}}_i(w)$ be the rightmost occurrence of $w$ over $T_i$, where $T_i$ is the prefix of the text $T$ ending at the position $i$ in $T$. Obviously, $\mbox{\emph{rightmost}}_n(w)=\mbox{\emph{rightmost}}(w)$ for $|T|=n$.
\end{definition}

In many practical data compression algorithms the text we are able to refer to is just a portion of the whole text. Let $T[j:i]$ be the factor of the text $T$ starting from the position $j$ and ending to the position $i$. We generalize the definition of $\mbox{\emph{rightmost}}(w)$ over a factor $T[j:i]$ of $T$ as follows.

\begin{definition} Let $\mbox{\emph{rightmost}}_{j,i}(w)$ be the rightmost occurrence of $w$ over $T[j:i]$, where $T[j:i]$ is the factor of the text $T$ starting at the position $j$ and ending at the position $i$ of length $i-j+1$. Obviously, $\mbox{\emph{rightmost}}_{1,n}(w)=\mbox{\emph{rightmost}}(w)$ for $|T|=n$.
\end{definition}

\begin{definition} \textbf{The Rightmost Equal-Cost Position (REP).}
Let us suppose that we have a text $T\in \Sigma^*$, a pattern $w\in \Sigma^*$ and a point $i$ in time. 
If $w$ appears at least once in $T_i$, then a \emph{rightmost equal-cost position} of $w$, REP$(w)$, over $T_i$ is a position $j$ in the text which offset is in $[\mbox{\emph{rightmost}}_i(w)]$, where $[\mbox{\emph{rightmost}}_i(w)]$ is the equivalence class induced by the relation \emph{having equal codeword bit-length} containing the element $\mbox{\emph{rightmost}}_i(w)$. 
Otherwise, when $w$ does not appear in $T_i$, the REP$(w)=0$.
\end{definition}

\section{The Multilayer Suffix Tree}\label{chapter:mlst}

The main idea of this new data structure, is based on two observations. 
The first one is that the equivalence relation \emph{having equal codeword bit-length} induces a partition on $\emph{Offset}(w)$, for any $w$, and also induces a partition on the set of all the possible offsets over a text $T$ independently from any specific factor, i.e. on the set $[1..|T|]$. 
The second observation is that for any encoding method for which the \emph{Bitlen} function is a monotonic function, each equivalence class in $[1..|T|]$ is composed by contiguous points in $[1..|T|]$. Indeed, given a point $p \in [1..|T|]$, the equivalence class $[p]$ is equal to the set $[j..i]$, with $j\leq p\leq i$, $j=\min\{x\in [p]\}$ and $i=\max\{x\in [p]\}$.

Assuming the \emph{Bitlen} function as the cost function of the dictionary phrases, we exploit the discreteness of \emph{Bitlen} in order to solve the REP problem by using a set of suffix trees for sliding window, each one devoted to one or more classes of equivalence of the relation \emph{having equal code bit-length}. 

\medskip

Fixed an encoding method for numbers and a text $T$, we assume that \emph{Bitlen} is a monotonic function and that we know the set $B=\{b_1, b_2, \ldots, b_s\}$, with $b_1<b_2<\cdots <b_s$, that is the set of the \emph{Bitlen} values of all the possible offsets over $T$.
We define the set $SW=\{sw_1, sw_2, \ldots , sw_s\}$, with $sw_1<sw_2<\cdots <sw_s$, where $sw_i$ is the greatest integer (smaller than or equal to the length of the text $T$) such that \emph{Bitlen}$(j)$ is less that or equal to $b_i$. More formally, $sw_i=\max\{j \leq |T| \ |\ \mbox{\emph{Bitlen}}(j)\leq b_i\}$. Notice that $sw_s=|T|$.

All the standard non-unary representation of numbers satisfy the following property.

\begin{property}\label{propertyexp}
There exists a constant $k>1$ and an integer $\hat{k}$ such that for any $\hat{k} \leq i < s$ one has $sw_i\geq k\, sw_{i-1}$.
\end{property}

The Huffman code do not strictly satisfy above property, but it easy to find a function $f(w)$ always greater than \emph{Bitlen}$($Huffman code of $w)$ that is a good approximation of it and it does have the Property \ref{propertyexp} or, alternatively, you can impose that property rearranging the Huffman trees.

\medskip

Let us consider an example where integers are encoded by using the  Elias $\gamma$ codes. 
The Table \ref{table:gamma} reports the Elias $\gamma$ codes and the \emph{Bitlen} values for integers from $1$ to $18$. 
Suppose, for instance, that our cost function is associated to this Elias codes and that  we have a text 
$T$ of length $18$. The set $B$ therefore is $B = \{b_1=1, b_2=3, b_3=5, b_4=7, b_5=9\}$ and 
we have that $sw_1=1$, $sw_2=3$, $sw_3=7$, $sw_4=15$ and $sw_5=18$. Notice that, indeed, 
Property \ref{propertyexp} is satisfied for $\hat{k}=2$ and  $k=2$.

\begin{table}[b!]
\begin{center}
\begin{tabular}{| c | c | c |}
\hline
$i$  &$\gamma(i)$ \qquad & \emph{Bitlen}$(\gamma(i))$  \cr
\hline
 1 & 1 & 1 \cr
 2 & 010 & 3 \cr
 3 & 011 & 3 \cr
  4 & 00100 & 5 \cr
  5 & 00101 & 5 \cr
  6 & 00110 & 5 \cr
  7 & 00111 & 5 \cr
  8 & 0001000 & 7 \cr
  9 & 0001001 & 7 \cr
  10 & 0001010 & 7 \cr
  11 & 0001011 & 7 \cr
  12 & 0001100 & 7 \cr
  13 & 0001101 & 7 \cr
  14 & 0001110 & 7 \cr
  15 & 0001111 & 7 \cr
  16 & 000010000 & 9 \cr
  17 & 000010001 & 9 \cr
  18 &  000010010  & 9 \cr
\hline
\end{tabular}
\end{center}
\vspace{0.20cm}\caption{Elias $\gamma$ code for integers from $1$ to $18$ and their \emph{Bitlen} value.}\label{table:gamma}
\end{table}

\medskip

Let us now introduce the Multilayer Suffix Tree data structure.
We suppose that a text $T$ is provided \textit{online} and at time $i$ the first $i$ characters of $T$ have been read, i.e. at time $i$ the prefix $T_i$ of length $i$ of the text $T$ has been read.

\begin{definition}
The Multilayer Suffix Tree is a data structure composed by the set $S=\{ S_{sw_{1}}, S_{sw_2}, \ldots , S_{sw_s} \}$ of suffix trees
where, for any $\alpha \in SW$ and at any moment $i$, 
$S_\alpha$ is the suffix tree for sliding window of $T_i$ with sliding window of size $\alpha$ such that $S_\alpha$ represents all the factors of $T[i - \alpha : i]$. We call $S_\alpha$ simply the layer $\alpha$ or the layer of size $\alpha$.
\end{definition}

From now on we will refer to suffix trees or layers indifferently.

We use the \textit{online} suffix tree for sliding window construction algorithm introduced by Larson in \cite{DBLP:conf/dcc/Larsson96} and later refined by Senft in \cite{SENFT05,10.1109/DCC.2006.11}, in order to build each layer of our multilayer suffix tree. 
Let us recall that in \cite{DBLP:conf/dcc/Larsson96} an \textit{online} and linear time construction for the suffix tree is reported. The suffix tree uses linear space w.r.t. the size of the sliding window.
Therefore, for any $S_\alpha \in S=\{ S_{sw_1}, S_{sw_2}, \ldots , S_{sw_s} \}$, $S_\alpha$ is a full-text index for $T[i - \alpha : i]$, where $T$ is a given text and $i$ is a point in time. $S_\alpha$ uses $O(\alpha)$ space. 

In order to decrease the practical space requirement, we think that it is possible to adapt our data structure to work with other classic indexes for sliding window  (see for instance \cite{inecdawg,Na:2003:TST:899776.899780,Senft:2008:SCP:1483948.1483960}) or to use the recently presented compressed versions of the suffix tree (e.g.
\cite{Grossi:2005:CSA:1093654.1096192,Sadakane:2007:CST:1326296.1326297}).

\begin{proposition}\label{pro:ricercabinaria}
$1$. If a pattern $w$ is in layer $\alpha$ with $\alpha \in SW$, then $w$ is also in layer $\beta$ for any $\beta \in SW$ with $\alpha \leq \beta$.
$2$. If a pattern $w$ is not in a layer $\alpha$, $\alpha \in SW$, then $w$ is not in layer $\beta$ with  $\beta\leq\alpha$
\end{proposition}
\begin{proof}
The proof of the property at point $1$ comes immediately from suffix trees properties. Since layer $\alpha$ is a full text index for $T[i - \alpha : i]$ and layer $beta$ is a full-text index for $T[i - \beta : i]$, for any $\alpha$, $\beta \in SW$ with $\alpha \leq \beta$ and for any $i$, $T[i - \alpha : i]$ is a suffix of $T[i - \beta : i]$. The property at point $2$ can be deduced by point $1$.
\end{proof}

\smallskip

\begin{proposition}
Fixed a text $T$ of size $|T|=n$, at any moment $i$ with $0\leq i \leq n$ and for any standard variable-length code, the multilayer suffix tree uses $O(i)$ space.
\end{proposition}
\begin{proof}
Since at time $i$ the maximum offset of all the occurrences in the text $T_i$ is $O(i)$, for any standard variable-length code the maximum value of the set $SW$ is $O(i)$. 
Since Property \ref{propertyexp}  holds for the set $SW$ and since the multilayer suffix tree space is equal to the sum of the space of its layers, as an immediate consequence we have that space used by the multilayer suffix tree is $O(\sum\limits_{\alpha \in SW} \alpha) = O(i)$.
\end{proof}

For instance, at time $i$, if we consider the usual binary representation of numbers, the  values $\alpha \in SW$ turn out to be powers of $2$ from $1$ to $j$, where $j$ is the greatest power of $2$ smaller than $i$, plus $i$, i.e. $SW=\{1,2,4,\dots j,i\}$ and $\sum\limits_{\alpha \in SW} \alpha < 2j+i < 3i$ and $O(\sum\limits_{\alpha \in SW} \alpha) = O(i)$. 

From the linear time of the \textit{online} construction of the suffix tree for sliding window and since the number of layers is $|SW|=O(\log |T|)$,  we can immediately state the following proposition.

\begin{proposition}
Given a text $T$ of length $|T|=n$, for any standard variable-length code, it is possible to build \textit{online} the multilayer suffix tree in $O(n \log n)$
amortized time.
\end{proposition}

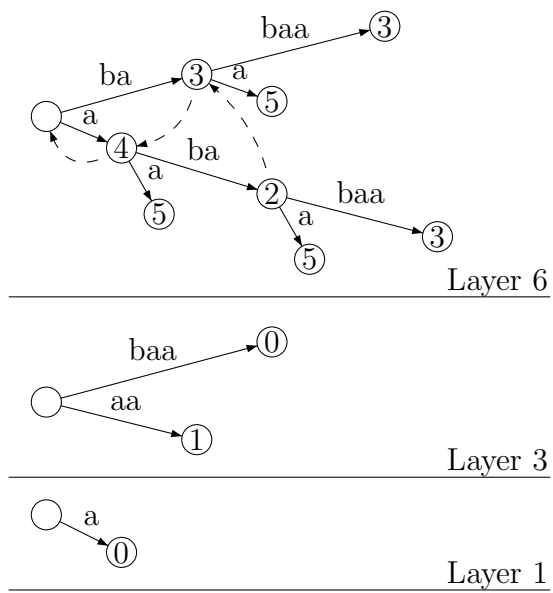
\begin{figure}[t!] 
\centering{
\begin{picture}(70,90)(0,-10)
\gasset{Nw=4,Nh=4}
   \node(1)(0,53){}
   \node(2)(10,48.8){$4$}
   \node(3)(30,42.7){$2$}
   \node(4)(52,37){$3$}
   \node(5)(20,58.6){$3$}
   \node(6)(45,65){$3$}
   \node(7)(15,40){$5$}
   \node(8)(35,34){$5$}
   \node(9)(30,55){$5$}
   \drawedge(1,2){a}
   \drawedge(2,3){ba}
   \drawedge(2,7){a}
   \drawedge(3,4){baa}
   \drawedge(3,8){a}
   \drawedge(1,5){ba}
   \drawedge(5,6){baa}
   \drawedge(5,9){a}
   \drawedge[AHnb=1,dash={1.5}{1.5},curvedepth=-2](3,5){}
   \drawedge[AHnb=1,dash={1.5}{1.5},curvedepth=3](5,2){}
   \drawedge[AHnb=1,dash={1.5}{1.5},curvedepth=4](2,1){}
   
   \node(11)(0,15){}
   \node(13)(20,10){$1$}
   \node(15)(30,23){$0$}
   \drawedge(11,13){aa}
   \drawedge(11,15){baa}

   \node(21)(0,0){}
   \node(22)(10,-5){$0$}
   \drawedge(21,22){a}
   
   \drawline[AHnb=0](-5,29)(67,29)
   \node[Nframe=n](L)(60,31){Layer $6$} 
   \drawline[AHnb=0](-5,5)(67,5)
   \node[Nframe=n](L)(60,7){Layer $3$} 
   \drawline[AHnb=0](-5,-10)(67,-10)
   \node[Nframe=n](L)(60,-8){Layer $1$} 
   
\end{picture}
\vspace{0.30cm}
\caption{
The Multilayer Suffix Tree for the text $T=\texttt{ababaa}$ and for the Elias $\gamma$-code, where $B=\{1,3,5\}$, $SW=\{1,3,6\}$.
The solid edges are regular ones and the dashed links are the suffix-links of internal nodes. For convenience, we added edge labels with the substring of $T$ associated to edges. The node value is the position over the text of the incoming edge. We omitted the string depth of nodes.
Let consider, for instance, the phrase $w=` ba"$ with 
\emph{Pos}$(w)=\{1,3\}$, \emph{Offset}$(w)=\{4,2\}$, 
\emph{rightmost}$(w)=2$ and $\gamma(2)=$\textit{010}.
Since $w$ is in layer $sw_2=3$ and is not  in layer $sw_1=1$, we have that \emph{Bitlen}$($\emph{rightmost}$(w))=3$ is equal to 
\emph{Bitlen}$(sw_2=3)=b_2=3$.
}
\label{fig:multigraph1}
}
\end{figure}

\section{Solving REP for a Generic Pattern}\label{chapter:REP(p)}

We want now to show how to answer to the REP$(p)$ for a given pattern $p$ of length $|p|$ 
and a text $T_i$ of length $i$ in $O(|p| \log \log i)$.

\begin{proposition}\label{pro:ricerca}
If a pattern $p$ is in layer $\beta$ and is not in layer $\alpha$, where $\alpha$ is the maximum of the values in $SW$ smaller than $\beta$, then
any occurrence of $p$ in the layer $\beta$ correctly solve the REP problem for the pattern $p$.
\end{proposition}

\begin{proof}
If a pattern $p$ is in layer $\beta$ and is not in layer $\alpha$, where $\alpha$ is the maximum of the values in $SW$ smaller than $\beta$, then from Prop. \ref{pro:ricercabinaria}, it follows that $\beta$ is the smallest layer where it appears $p$. Therefore $p$ has at least one occurrence in $T[i - \beta : i-\alpha-1]$, i.e. $\mbox{\emph{rightmost}}_i(p)\in (\alpha .. \beta]$. 
Since $(\alpha .. \beta]$ is the equivalence class $[\beta]$ of the \emph{having equal codeword bit-length} relation, we have that any occurrence of $p$ in the layer $\beta$ correctly solve the REP problem for $p$.
\end{proof}

\begin{remark}
Let us notice that if $S_{sw_x}$ is the smallest layer where $\mbox{\emph{rightmost}}_i(w)$ appears, then the $\mbox{\emph{Bitlen}}$ value of the offset of $\mbox{\emph{rightmost}}_i(w)$ is equal to $b_x$. 
\end{remark}

Using above proposition, we are able to solve the problem REP$(p)$ once we find the smallest layer containing the rightmost occurrence of the pattern, if any, otherwise we just report $0$.

What follows is the trivial search of the smallest layer that contains an occurrence of a given pattern.

Given a pattern $p$ at time $i$, we look for $p$ in the layer $sw_1$, i.e. the smallest layer.  
If $p$ is in $S_{sw_1}$, then all the occurrences of $p$ in $T[i- sw_1 :i]$ belong to the class of the rightmost occurrence of $p$ over $T$. 
If $p$ is not in $S_{sw_1}$, then we look for any occurrence of $p$ is in $S_{sw_2}$, the next layer in increasing order. 
If $p$ is in $S_{sw_2}$, since it is not in $S_{sw_1}$, for the Prop. \ref{pro:ricerca}, any occurrence of $p$ in this layer belong to the rightmost occurrence of $p$ over $T_i$. 
Continuing in this way, as soon as we found an occurrence of $p$ in a layer, this occurrence correctly answer to REP$(p)$. By using this trivial approach we can find 
REP$(p)$ in time proportional to the length of the pattern times the number of layers.

Since many of the classic variable-length codes for integers, like the Elias's $\gamma$-codes, produce codewords of length proportional to the logarithm of the represented value, we can assume that the cardinality of $SW$ is $O(\log |T|)$.
Since that $|T_i|=i$, in the \emph{online} fashion, we have that the time for REP$(p)$ is $O(|p| \log i)$.

A similar result can be obtained by using a variant of the Amir et al. algorithm presented \cite{Amir_onlinetime}, but it does not support REP(LPF) operations in constant time and it does not support sliding window over the text.

\medskip

Since Prop. \ref{pro:ricercabinaria} holds for the layers of our data structure, we can use the binary search to find the smallest layer containing a given pattern. Since $|SW|=O(\log i)$, for any classic variable-length code,  the number of layers in our structure is $O(\log i)$ and the proof of following proposition comes straightforward. 

\begin{proposition}
Using any classic variable-length code, at any time $i$ the multilayer suffix tree is able to answer to REP$(pattern)$ for a given pattern in $O(|\mbox{\emph{pattern}}| \log\log i)$ time.
\end{proposition}

\section{Solving REP(LPF) and REP(SPF)}\label{chapter:REP(LPF)}

As pointed out in the Senft's paper \cite{10.1109/DCC.2006.11}, the online suffix tree construction algorithms in \cite{Fiala:1989:DCF:63334.63341,DBLP:conf/dcc/Larsson96,UK95}) keep track of the longest repeated suffix of the text. 
We use those information, one from any layer of the multilayer suffix tree, to answer in constant time to the REP$($LPF$)$ problem and the  REP(SPF), where LPF is the longest previous factor and SPF is the set of prefixes of LPF.

Recall that, using Ukkonen's nomenclature, an Implicit Suffix Tree $I_k$, $1<k<n$, for the text $T=a_1\dots a_n$ is a Suffix Tree for the text $T[1..k]=T_k$. It is built starting from $I_{k-1}$ by implicitly extending all the suffixes of $T_{k-1}$ that are unique, i.e. those that are leaves, and explicitly adding the suffixes of $T_k$ in decreasing order down to the longest suffix of $T_k$ that is already in $I_{k-1}$, i.e. the longest suffix of $T_k$ that appears at least twice in $T_k$. In order to efficiently perform the explicit extensions, the construction algorithm maintains at each step a reference to the longest repeated suffix of the text. 
Let us say that at the end of the step $k$ 
the longest repeated factor is $T[i \dots k]$. Eventually in some successive steps the longest repeated factor became $T[i+1 \dots k']$.

In order to solve the REP$($LPF$)$ problem in constant time, we let independently to grow each layer of the multilayer suffix tree to $I_{(k'-1)}$, where, for any layer, $k'-1$ is the last step of the construction algorithm which longest repeated factor starts in position $i$. This does not change the time complexity of the algorithm, since the overall steps for a text $T$ are the same.
In the mean time, in order to correctly represent the factors of offset $b_\alpha$ in each layer $\alpha$, the sliding widow size are dynamically set to $SW_\alpha + \mbox{LPF}_\alpha(i)$.

Now that any layer point out to the longest repeated suffix starting in position $i$ accordingly to the layer size, we just take in constant time, along the building process, from any layer one occurrence of its longest repeated suffix. Then, we keep the value coming from the smallest layer with a pointed factor of length equal to the length of the biggest layer, i.e. LPF$(i)$. 

\begin{proposition}
Given a text $T$ of length $n$ in the dynamic setting, at time $i$, the multilayer suffix tree of the text $T[1..i+$LPF$[i]]$ is able to solve REP$($LPF$)$ in constant time.
\end{proposition}

Since LPF$(i)$ is equal to the longest repeated suffix of $T[1..i+$LPF$[i]]$ by the very definition of LPF, the proof of above proposition come straightforward from Proposition \ref{pro:ricercabinaria}.

Concerning the occurrences of the set of prefixes of LPF (SPF for short), we are interested to those having a different bit length of their rightmost position, as the missing ones are just prefix of an occurrence of a longer one.
With arguments similar to the ones used for the LPF case, we maintain inside the multilayer suffix tree a list SPF of $($length, offset$)$ pairs in the following way.
At any time $i$, let $occ$ be the position of LPF$(i)$ found in the largest layer. The pair $($LPF$[i],occ)$ is added to the empty list SPF.
Let LPF$_\alpha$ and $occ_\alpha$ be the length and the offset of the longest repeated suffix pointed out by the layer $\alpha$.
For any $\alpha$ in decreasing order if the last pair in the list SPF has length value equal to LPF$_\alpha$, then this element is updated with the offset value $occ_\alpha$. Otherwise the pair $($LPF$_\alpha$, $occ_\alpha)$ is appended to the list. Notice that the size of the list is smaller that or equal to the number of layers $O(\log i)$ as already assumed. In order to solve the REP(SPF) problem it suffice to retrieve a pointer to the internal SPF list of the multilayer suffix tree. 

The proof of the following proposition is straightforward.

\begin{proposition}
Given a text $T$ of length $n$ in the dynamic setting, at time $i$, the multilayer suffix tree of the text $T[1..i+$LPF$[i]]$ is able to solve REP$($SPF$)$ in constant time.
\end{proposition}

\smallskip

At any time $i$, since any REP$($SPF$)$ value has \emph{equal-cost} w.r.t. the rightmost occurrence of the corresponding LPF prefix, REP$($SPF$)$ can be used  as dictionary pointers for the dictionary phrases matching the factors of the text starting at position $i$ in any  LZ77 based compression algorithm.
Due to space constrain we omit the details of the next proposition. 
\begin{proposition}
Given a text $T$, using at any time $i$ the REP$($SPF$)$ values provided by the multilayer suffix tree, one can correctly build the parsing graph $G'_{A,T}$ of the \emph{Dictionary-Symbolwise Flexible Parsing} algorithm.
\end{proposition}

\bigskip

\begin{table}[b!]
\begin{center}
\begin{tabular}{| l | r | r | r |c|}
\hline
Input  \quad  & \quad  Length \quad & \quad  MLST  \quad & \quad  RMST \quad & \quad  $\Delta$ \quad \cr
\hline
bible			&	$4$ MB 	& $8.52\ \mu s$ & $1.77\ \mu s$ &  $4.81$  \cr 	
enwik8		& $100$ MB & $16.15\ \mu s$ & $3.00\ \mu s$ &  $5.38$\cr
dna			& $104$ MB & $8.51\ \mu s$ & $2.19\ \mu s$ &  $3.88$\cr

\hline
\end{tabular}
\end{center}
\vspace{0.20cm}\caption{Execution time table for different input. The bible file belongs to the Large Canterbury Corpus, the enwk8 file contains the first $100$ MB of the English Wikipedia database 
and the dna file belongs to the Repetitive Pizza\&Chili Corpus.
The MLST column report the total time per character of the construction of the Multilayer Suffix Tree with max window size of $4$ MB. 
The RMST column report the building time per character of one single suffix tree with sliding window of $4$ MB.
The $\Delta$ column contains the ratio between MLST and RMST times.
}\label{table:time}
\end{table}

We present some experimental results in order to evaluate the practical requirements in terms of space and time. We used as cost function just the length of the binary representation of numbers and then sliding windows of size equal to the increasing powers of $2$ up to a fixed length \emph{MaxDictionarySize} equal to $2^{24}$. 
The space grows up to $35$ Bytes times \emph{MaxDictionarySize} for regular text file like the bible and the enwik8 ones (see Table \ref{table:time} caption for more details), as our implementation uses $17$ Byte in average per character for any layer. 
We compared the multilayer suffix tree (MLST) running time with the running time of 
a suffix tree with sliding window of size \emph{MaxDictionarySize} maintaining the rightmost occurrence of any internal node (RMST). At any time $i$, the position value of internal nodes in the path from the root to the point of insertion of the new leaf are updated.
The average running time of the MLST is about $5$ times greater than the RMST time, as shown in Table \ref{table:time}. 
Furthermore, since all the layers in the multilayer suffix tree are independent each other, it is easy to speed up the overall time by using a parallel handling of the layers.  
Moreover, since maintaining the rightmost position in RMST runs in $O(n \log n)$ average case and $O(n^2 )$ in the worst case while the MLST run time is $O(n \log n)$ in the  worst case, the use of the MLST is encouraged when the dictionary size become large and/or the text contains highly repetitive factors, e.g. using molecular biology data, as the $\Delta$ value for the dna file suggests.

\bibliographystyle{abbrv}
\bibliography{compression}

\begin{thebibliography}{10}

\bibitem{Amir_onlinetime}
A.~Amir, G.~M. Landau, and E.~Ukkonen.
\newblock Online timestamped text indexing.
\newblock {\em Information Processing Letters}, 82(5):253 -- 259, 2002.

\bibitem{DBLP:conf/dcc/CohnK96}
M.~Cohn and R.~Khazan.
\newblock Parsing with prefix and suffix dictionaries.
\newblock In J.~A. Storer and M.~Cohn, editors, {\em Data Compression
  Conference}, pages 180 -- 189. IEEE Computer Society, 1996.

\bibitem{cglmr_iwoca10}
M.~Crochemore, L.~Giambruno, A.~Langiu, F.~Mignosi, and A.~Restivo.
\newblock Dictionary-symbolwise flexible parsing.
\newblock In {\em IWOCA'2010}, volume 6460 of {\em Lecture Notes in Computer
  Science}, pages 390--403, 2011.

\bibitem{cglmr_JDA2011}
M.~Crochemore, L.~Giambruno, A.~Langiu, F.~Mignosi, and A.~Restivo.
\newblock Dictionary-symbolwise flexible parsing.
\newblock {\em Journal of Discrete Algorithms - IWOCA'10 Special Issue}, 2011.

\bibitem{FerraginaSODA09}
P.~Ferragina, I.~Nitto, and R.~Venturini.
\newblock On the bit-complexity of {L}empel-{Z}iv compression.
\newblock In {\em SODA '09}, pages 768--777. Society for Industrial and Applied
  Mathematics, 2009.

\bibitem{Fiala:1989:DCF:63334.63341}
E.~R. Fiala and D.~H. Greene.
\newblock Data compression with finite windows.
\newblock {\em Commun. ACM}, 32:490--505, April 1989.

\bibitem{Grossi:2005:CSA:1093654.1096192}
R.~Grossi and J.~S. Vitter.
\newblock Compressed suffix arrays and suffix trees with applications to text
  indexing and string matching.
\newblock {\em SIAM J. Comput.}, 35(2):378--407, Aug. 2005.

\bibitem{DBLP:conf/dcc/Horspool95}
R.~N. Horspool.
\newblock The effect of non-greedy parsing in {Z}iv-{L}empel compression
  methods.
\newblock In {\em Data Compression Conference}, pages 302--311, 1995.

\bibitem{inecdawg}
S.~Inenaga, A.~Shinohara, M.~Takeda, and S.~Arikawa.
\newblock Compact directed acyclic word graphs for a sliding window.
\newblock In {\em SPIRE}, pages 310--324, 2002.

\bibitem{cancan}
T.~Y. Kim and T.~Kim.
\newblock On-line optimal parsing in dictionary-based coding adaptive.
\newblock {\em Electronic Letters}, 34(11):1071--1072, 1998.

\bibitem{DBLP:conf/dcc/Larsson96}
N.~J. Larsson.
\newblock Extended application of suffix trees to data compression.
\newblock In {\em Data Compression Conference}, pages 190--199, 1996.

\bibitem{Lempel:2006:AOP:2263298.2268473}
A.~Lempel, S.~Even, and M.~Cohn.
\newblock An algorithm for optimal prefix parsing of a noiseless and memoryless
  channel.
\newblock {\em IEEE Trans. Inf. Theor.}, 19(2):208--214, Sept. 2006.

\bibitem{DBLP:conf/soda/MatiasS99}
Y.~Matias and S.~C. Sahinalp.
\newblock On the optimality of parsing in dynamic dictionary based data
  compression.
\newblock In {\em SODA}, pages 943--944, 1999.

\bibitem{Na:2003:TST:899776.899780}
J.~C. Na, A.~Apostolico, C.~S. Iliopoulos, and K.~Park.
\newblock Truncated suffix trees and their application to data compression.
\newblock {\em Theor. Comput. Sci.}, 304:87--101, July 2003.

\bibitem{Sadakane:2007:CST:1326296.1326297}
K.~Sadakane.
\newblock Compressed suffix trees with full functionality.
\newblock {\em Theor. Comp. Sys.}, 41(4):589--607, Dec. 2007.

\bibitem{books/Salomon:2007:DC}
D.~Salomon.
\newblock {\em Data compression - The Complete Reference, 4th Edition}.
\newblock Springer, 2007.

\bibitem{SENFT05}
M.~Senft.
\newblock Suffix tree for a sliding window: An overview.
\newblock In {\em WDS'05}, pages 41--46, 2005.

\bibitem{10.1109/DCC.2006.11}
M.~Senft.
\newblock Compressed by the suffix tree.
\newblock {\em DCC'06, IEEE Computer Society}, 0:183--192, 2006.

\bibitem{Senft:2008:SCP:1483948.1483960}
M.~Senft and T.~Dvo\v{r}\'{a}k.
\newblock Sliding cdawg perfection.
\newblock In {\em SPIRE '08, Lecture Notes in Computer Science}, pages
  109--120. Springer-Verlag, 2009.

\bibitem{DBLP:journals/jacm/StorerS82}
J.~A. Storer and T.~G. Szymanski.
\newblock Data compression via textural substitution.
\newblock {\em J. ACM}, 29(4):928--951, 1982.

\bibitem{UK95}
E.~Ukkonen.
\newblock On-line construction of suffix-trees.
\newblock {\em Algorithmica}, 14:249--260, 1995.

\bibitem{lz77}
J.~Ziv and A.~Lempel.
\newblock A universal algorithm for sequential data compression.
\newblock {\em IEEE Transactions on Information Theory}, 23(3):337--343, 1977.

\end{thebibliography}

\end{document}